\documentclass[11pt]{article}
\usepackage{amssymb,amsmath,amsthm}
\usepackage{xcolor}
\definecolor{darkred}{HTML}{800000}

\usepackage[margin=1.03in]{geometry}
\usepackage{titlesec}
\usepackage[pagebackref,bookmarks=true,pdftex]{hyperref}
\hypersetup{
  colorlinks=true,
  linkcolor=darkred,
  urlcolor=darkred,
  citecolor=darkred,
  linktocpage=true,
}
\usepackage[nameinlink,capitalize,noabbrev]{cleveref}
\crefname{equation}{equation}{equations}
\usepackage{microtype}
\makeatletter
\let\originaltagform@\tagform@
\renewcommand{\tagform@}[1]{\textcolor{darkred}{\originaltagform@{#1}}}
\makeatother

\newtheoremstyle{boldnames}
  {\topsep}{\topsep}{\itshape}{}{\bfseries}{.}{.5em}
  {\thmname{#1}\thmnumber{ #2}\thmnote{ \textbf{(#3)}}}
\theoremstyle{boldnames}
\newtheorem{theorem}{Theorem}
\newtheorem{maintheorem}[theorem]{Main Theorem}
\newtheorem{lemma}[theorem]{Lemma}
\newtheorem{corollary}[theorem]{Corollary}
\newtheorem{proposition}[theorem]{Proposition}
\crefname{maintheorem}{main theorem}{main theorems}
\Crefname{maintheorem}{Main Theorem}{Main Theorems}
\crefname{lemma}{Lemma}{Lemmas}
\crefname{proposition}{Proposition}{Propositions}
\Crefname{proposition}{Proposition}{Propositions}
\Crefname{appendix}{Appendix}{Appendices}

\renewcommand{\leq}{\leqslant}
\renewcommand{\geq}{\geqslant}

\newcommand{\F}{\mathbb F_2}
\newcommand{\wt}{\operatorname{wt}}
\newcommand{\swt}{\operatorname{swt}}
\newcommand{\dist}{\operatorname{dist}}
\newcommand{\supp}{\operatorname{supp}}
\newcommand{\ones}{\mathbf 1}
\newcommand{\defeq}{\mathrel{\overset{\mathrm{def}}{=}}}

\begin{document}

\title{Metric Self-Dual Completion and
Optimal Additive Hardness \\ for Quantum and Graph-State Distance}
\author{Rafail Ostrovsky\\ \ \\
UCLA}
\date{\today}

\setcounter{page}{0}
\maketitle
\thispagestyle{empty}

%%%%%%%%%%%%%%%%%%%%%%%%%%%%%%%%%%%%%%%%%%%%%%%%%%%%%%%%%%%%

\begin{abstract}
We prove that the quantum code distance is NP-hard to approximate
within an additive error of $c N$, for some constant $c >0$,
where $N$ is the number of qubits.
Our reductions are deterministic.
This improves the previous square-root additive gap to $\Omega(N)$ and resolves the explicitly stated linear-gap question of Kapshikar and Kundu.
Our result holds for CSS codes with identical
$X$- and $Z$-check spaces, and with a constant rate and constant relative
distance.
For every fixed $\lambda>1$, there is a constant $c>0$ such that hardness still holds even when every
nonidentity stabilizer has weight greater than $\lambda$ times the quantum
distance.
We also improve the hardness gap of graph state distance on $N$ vertices of Grigorescu, Jha, and Samperton from cube-root to $\Omega(N)$, resolving their explicitly stated open question. Both hardness results are asymptotically optimal since both distances are at most $N$. Our graph state distance hardness result holds for balanced bipartite
graphs with a binary adjacency matrix that is its own inverse (mod 2).

\smallskip
Our main technique  for both hardness bounds above is classical: we show how to convert any code $C$ of length $m$ into a self-dual code $A(C)$ of length $N=\Theta(m)$
while exactly doubling the original coset metric.
The conversion is deterministic and efficient.
We call it the {\em metric self-dual completion} of $C$. It comes with a linear embedding $\tau: \mathbb F_2^m \hookrightarrow \mathbb F_2^N$.
The embedding doubles all Hamming
distances between vectors in $\mathbb F_2^m$ and all pairwise distances between corresponding cosets.
The embedding also guarantees that all codewords of $A(C)$ of weight at most $2m$ are exactly  $\tau(C)$.
\end{abstract}

%%%%%%%%%%%%%%%%%%%%%%%%%%%%%%%%%%%%%%%%%%%%%%%%%%%%%%%%
\newpage
\pagenumbering{roman}
\tableofcontents
\newpage
\pagenumbering{arabic}
%%%%%%%%%%%%%%%%%%%%%%%%%%%%%%%%%%%%%%%%%%%%%%%%%%%%%%%%%%%%%

%%%%%%%%%%%%%%%%%%%%%%%%%%%%%%%%%%%%%%%%%
%%%%%%%%%%%%%%  SECTION 1 %%%%%%%%%%%%%%%%%%
%%%%%%%%%%%%%%%%%%%%%%%%%%%%%%%%%%%%%%%%%%%%%

\section{Introduction}
Recent work established  $N^{1/2}$ additive hardness gap for
quantum-code distance~\cite{KapshikarKundu23} and $N^{1/3}$ hardness gap for graph-state distance~\cite{GJS25}.
These works left open the tantalizing question whether either problem admits a linear additive
hardness gap.
We answer both questions in the affirmative.

Our main construction is classical. We transform any nonzero binary
code into a self-dual code with linear expansion in length while exactly doubling its minimum
distance. Our construction also preserves all distances
between corresponding source cosets. These additional metric guarantees hold for every
input nonzero code, independently of the hardness applications.
The main idea is to begin with a (separate) self-dual code whose nonzero words have large
weight and combine it in a novel way with the input code we are operating on. We retain the words orthogonal to an embedded copy of the input code,
then add that copy. We show that the dimensions removed and later added are equal, so the new
code is self-dual. Our weight separation guarantee ensures that every short word
comes from the input code. The resulting subspace and the completed code (from the two codes) give us the
quantum and graph-state additive hardness results, respectively. Since the output length
is linear,  both constructions preserve a linear hardness gap. 
All classical codes considered in this paper are binary linear codes. 

%%%%%%%%%%%%%%%%%%%%%%%%%%%%%%%%%%%%%%%%%%%%%%%%%%%%
\subsection{Metric Self-Dual Completion}

As usual,  $d(C)$ denotes the minimum distance of
$C$. Let $\dist(y,C)$ be the Hamming distance from $y$ to $C$.
Our main result (clauses~(\ref{main-clause-1})-(\ref{main-clause-4})) applies to every code $C\leq\mathbb F_2^m$.
For fixed $m$ and $H$, the same output length and embedding $\tau$
work for every $C$.
To simplify the notation, we refer to the self-dual completion of $C$
as $A$, instead of $A(C)$.
\begin{maintheorem}[Metric self-dual completion]
\label{thm-main}
$\forall$  $m\geq1$ and $H\geq2m$  $\exists$
$N=\Theta(H)$ and linear embedding $\tau:\F^m\hookrightarrow\F^N$, such that $\forall$ $C\leq\F^m$ $\exists$ self-dual code
$A=A^\perp\leq\F^N$ so that:
\begin{enumerate}
\renewcommand{\theenumi}{\roman{enumi}}
\renewcommand{\labelenumi}{(\theenumi)}
\item\label{main-clause-1}
$\forall x\in\F^m$, $2\cdot\wt(x)=\wt(\tau(x))$;
\item\label{main-clause-2}
$\forall x\in\F^m$,  $\dist(\tau(x),A)=2\cdot\dist(x,C)$;
\item\label{main-clause-3}
$\tau(C)=\{a\in A:\wt(a)\leq H\}$;
\item\label{main-clause-4}
A generator matrix for $A$ and  matrix for $\tau$ can be computed
in deterministic {\em poly}$(N)$ time.
\end{enumerate}
Furthermore, clause~(\ref{main-clause-2}) implies that the induced map
$y+C\mapsto\tau(y)+A$ is well-defined and injective.
Specifically, this correspondence doubles every pairwise coset distance.
Together with the self-dual code $A$, our construction produces a
self-orthogonal subspace $K\leq A$ where $\dim K=N/2-\dim C$ and
$d(K)>2H$. For $C\neq0$, $K$ satisfies:
$$
 \min_{v\in K^\perp\setminus K}\wt(v)=2d(C)
$$
\end{maintheorem}

By setting $H=2m$, we get a self-dual code $A$ of linear length (in the length of the original codeword) and, by clause~(\ref{main-clause-4}), do it in deterministic poly-time.
All codewords of $A$ of weight at most $2m$ are
exactly the embedded codewords of $\tau(C)$ by clause~(\ref{main-clause-3}).
That means that $d(A)=2d(C)$ whenever $C\neq0$, by clauses~(\ref{main-clause-1}) and~(\ref{main-clause-3}).
Thus, we obtain self-duality (with only a constant-factor blow-up) in the code length,
 while preserving the original coset geometry.
Our main (classical) construction and its proof appear in \cref{completion}.

%%%%%%%%%%%%%%%%%%%%%%%%%%%%%%%%%%%%%%%%%%%%%%%%%%%%%%%%%%

\paragraph{Two outputs of our main construction.}

Jumping ahead, $A$ will be used for graph-state distance hardness, and $K$ will be used as
both CSS check spaces. As we explain in
\cref{self-dual-exchange}, our main construction gives both hardness
results. We remark that the full source-coset metric isomorphism is a stronger guarantee than
these applications require. Generalizations of clauses~(\ref{main-clause-1})-(\ref{main-clause-4}) of our main theorem appear in the appendix.

%%%%%%%%%%%%%%%%%%%%%%%%%%%%%%%%%%%%%%%%%%%%%%%%%%%%%

\subsection{Optimal additive hardness for quantum and graph-state distance}
Recall that a distance problem has a \emph{linear additive gap} if, for some
constant $c>0$, it is NP-hard to distinguish the YES case
$f\leq T$ from the NO case $f>T+cN$. Such a gap implies NP-hardness of approximation within additive
error $c_0N$ for some constant $c_0>0$.
Here $f$ is the distance, $N$ is the number of qubits or vertices,
and the integer threshold $T$ is part of the input.  Inputs come from a ``promise problem'' and are guaranteed to satisfy one of the two conditions.
We use many-one deterministic Karp reductions to prove NP-hardness.

\paragraph{Distance of a quantum code.}
Quantum distance is the smallest number of qubits on which an
undetectable error can change the encoded information.
We prove hardness for Calderbank-Shor-Steane (CSS) codes with
identical $X$- and $Z$-check spaces.  For such a code
$Q=CSS(B,B)$, write $N$ for the number of qubits,
$k$ for the number of encoded (logical) qubits, and $d(Q)$ for the logical
distance of the quantum code.  We write $d(B)$ for the minimum
Hamming weight of a nonzero vector in $B$; equivalently, this is the minimum
weight of a nonzero stabilizer label.  See \Cref{quantum-preliminaries} for a purely
linear-algebra formulation of these notions.

%%%%%%%%%%%%%%%%%%%%%%%%%%%%%%%%%%%%%%%%%%%%%%%%%%%%
%%%%%%%%%%%%%%%%%%%%%%%%%%%%%%%%%%%%%%%%%%%%%%%%%%%%

For a fixed constant $c>0$, the additive-gap quantum-distance promise
problem takes as input an $N$-qubit code
$Q=\operatorname{CSS}(B,B)$ and an integer threshold $T$, with the promise that either
 $d(Q)\leq T$ or
 $d(Q)>T+cN$. The challenge is to determine which of these two conditions holds.

%%%%%%%%%%%%%%%%%%%%%%%
\begin{theorem}[Optimal additive hardness for quantum distance]
\label{thm-intro-quantum-hardness}
For every fixed $\lambda>1$, there are positive constants $c,r,\delta$
with the following property.  The promise problem above is NP-hard even when
$Q$ has at least $rN$ logical qubits, has distance at least $\delta N$,
and every nontrivial stabilizer has weight greater than $\lambda \cdot d(Q)$.
\end{theorem}

Our theorem shows that this problem remains hard for codes of constant rate and
constant relative distance, even when every nontrivial stabilizer has weight greater than $\lambda$ times the logical distance.
We stress that the construction does not guarantee sparse check matrices.  Thus, the above
theorem does not establish hardness for quantum LDPC codes.

Optimal additive hardness for quantum distance has received some attention in the recent literature. Specifically,
Kapshikar and Kundu~\cite[Corollary~2 and Section~V]{KapshikarKundu23}
established a randomized reduction square-root additive hardness gap and asked whether a linear gap could be obtained.
Grigorescu, Jha, and Samperton~\cite[Theorem~1 and Sections~1.1-1.2]{GJS25}
provided an alternative proof of the square-root hardness under deterministic Karp reductions and posed the linear-gap question as a main motivation.

Grigorescu, Jha, and Samperton
discuss a limitation of their hypergraph-product approach under
deterministic many-one reductions~\cite[Proposition~22]{GJS25}.
Our construction uses a self-dual host and a subspace exchange to obtain
linear output length.
Our \Cref{thm-intro-quantum-hardness} resolves the linear-gap
question posed in both papers above, going all the way  from $N^{1/2}$ to
$\Omega(N)$ with a deterministic Karp reduction.

%%%%%%%%%%%%%%%%%%%%%%%%%%%%%%%%%%%%%%
%%%%%%%%%%%%%%%%%%%%%%%%%%%%%%%%%%%%%%
%%%%%%%%%%%%%%%%%%%%%%%%%%%%%%%%%%%%%%
\paragraph{Graph-state Distance.}

We use the graph-state distance definition of Kovalev, Dumer, and Pryadko~\cite{KovalevDumerPryadko11}.
For a simple graph $G$, with adjacency matrix $\Gamma$, its graph-state distance $d(\Gamma)$ is the minimum weight of a
nontrivial element of the stabilizer of the graph state associated with $G$.
We provide a linear-algebraic definition of this graph-state distance in \cref{graph-preliminaries}.
Grigorescu, Jha, and Samperton~\cite[Theorem~2 and Section~5.3]{GJS25}
established an $N^{1/3}$ additive gap for this problem and asked whether it could be made linear.
Our \cref{graph-gap} resolves their question. Furthermore, it answers their question  even when $G$ is
bipartite and balanced,  and when  its adjacency matrix is an involution (over $\F$):

\begin{theorem}[Optimal additive hardness for graph-state distance]
\label{graph-gap}
$\exists$ $c,\delta>0$ s.t. Graph-state minimum distance is NP-hard to approximate within an additive error of $c\cdot N$ even when $G$ is bipartite and balanced and its adjacency matrix $\Gamma$
 satisfies $\Gamma^2=I$ and  $d(\Gamma) \geq\delta N$.
\end{theorem}

We start with Grigorescu et al.'s code-to-graph mapping~\cite[proof of Theorem~2]{GJS25}. We then use our
 metric self-dual completion as the missing piece: our construction doubles the source distance and produces a self-dual input.  Self-duality forces the graph to be balanced and gives $\Gamma^2=I$ as needed.

%%%%%%%%%%%%%%%%%%%%%%%%%%%%%%%%%%%%%%

\begin{samepage}
\begin{center}
\begin{tabular}{c|c|c}
 problem & previous additive gap & this paper\\ \hline
 CSS quantum distance & $N^{1/2}$ & $\Omega(N)$\\
 graph-state distance & $N^{1/3}$ & $\Omega(N)$
\end{tabular}
\end{center}
Since distances are at most $N$, our additive gaps are asymptotically optimal for both problems.
\end{samepage}

%%%%%%%%%%%%%%%%%%%%%%%%%%%%%%%%%%%%%%%
%%%%%%%%%%%%%%%%%%%%%%%%%%%%%%%%%%%%%%%
%%%%%%%%%%%%%%%%%%%%%%%%%%%%%%%%%%%%%%%%
\subsection{Classical consequences and related work}
\label{classical-consequences}

We first explore the implications of clauses~(\ref{main-clause-1}), (\ref{main-clause-3}), and~(\ref{main-clause-4}) of our main theorem (of metric self-dual completion).  We start with the result of
Dumer, Micciancio, and Sudan~\cite{DMS03} who showed
a linear additive (hardness)  gap for the minimum-distance code, using a randomized reduction. For multiplicative approximation, Cheng and Wan~\cite{CW12}
proved hardness within every fixed constant factor under deterministic reductions. The same result  was later shown using a PCP-free proof by
Bhattiprolu, Guruswami, Lee, and Ren~\cite{BGLR25}.
For codes of constant rate and relative distance, Austrin and
Khot~\cite{AustrinKhot14} obtained a deterministic constant-factor
hardness.

\medskip

We stress that none of these reductions guarantee self-duality of the output codes,
as explicitly discussed in Kapshikar and Kundu (in connection to quantum
applications, see~\cite[Section~I]{KapshikarKundu23}).
Our main result, specifically clauses~(\ref{main-clause-1}), (\ref{main-clause-3}), and~(\ref{main-clause-4}), allows us to adopt these classical hardness results to self-dual codes.

%%%%%%%%%%%%%%%%%%%%%%%%%%%%%%%%%%%%%%%%%%%%%%%%%%%%%%%%%%%%

Austrin and Khot~\cite[Theorem~1.1]{AustrinKhot14} showed that,
for some constants $\varepsilon>0$ and $\gamma>1$, it is
NP-hard under a deterministic reduction to distinguish
$d(C)\leq t$ from $d(C)>\gamma t$, where $t>0$ is an integer
and $C\leq\F^m$ satisfies $\dim C\geq\varepsilon m$ and
$d(C)\geq\varepsilon m$. Combining this with clauses~(\ref{main-clause-1}), (\ref{main-clause-3}), and~(\ref{main-clause-4}) of our main theorem gives:

\begin{corollary}[Optimal additive hardness for self-dual distance]
\label[corollary]{self-dual-distance-hardness}
$\exists$ $c,\delta>0$ such that, given a self-dual
$A\leq\F^N$ and threshold $T$, it is NP-hard (under
Karp reductions) to distinguish $d(A)\leq T$
vs. $d(A)>T+cN$, even when $d(A)\geq\delta N$.
\end{corollary}

\begin{proof}
Use Austrin-Khot hard instance $(C,t)$ as input to
\cref{thm-main}, using clauses~(\ref{main-clause-1}), (\ref{main-clause-3}), and~(\ref{main-clause-4}), with $H=2m$. For some fixed  constant $\alpha > 0$ observe that
$d(A)=2d(C)$ and $N\leq\alpha m$.
Assign
$T=2t$,
$\delta=2\varepsilon/\alpha$, and
$c=\delta(1-1/\gamma)$. We therefore have $d(A)\geq\delta N$.
The YES case gives $d(A)\leq T$. The NO case gives $d(A)>\gamma T$, thus:
$$
 d(A)-T>\left(1-\frac1\gamma\right)d(A)\geq cN
$$
\end{proof}

The hard instances of
\Cref{self-dual-distance-hardness} have a rate of $1/2$
and a relative distance of at least $\delta$.
%
%%%%%%%%%%%%%%%%%%%%%%%%%%%%%%%%%%%%%%%%%%%%%%%%%%%%%%%%%
%
Next, observe that clauses~(\ref{main-clause-1}) and~(\ref{main-clause-3}) of \cref{thm-main} give $d(A)=2d(C)$, which preserves every multiplicative
approximation gap when the threshold is also doubled.
Hence, when we apply the PCP-free hardness
result~\cite[Theorem~1.3]{BGLR25}, this gives the following:
\begin{corollary}[Hardness within every constant factor for self-dual distance]
For every fixed constant $\gamma>1$, given a self-dual $A\leq\F^N$ and a
threshold $T$, it is hard (under deterministic Karp reductions) to distinguish
$d(A)\leq T$ from $d(A)>\gamma T$.
This hardness proof is PCP-free.
\end{corollary}

\begin{proof}
Choose $\gamma'>\gamma$. We reduce a nonzero
instance $(C,t)$ of Bhattiprolu, Guruswami, Lee, and Ren,
with $C\leq\F^m$, $t>0$, and hardness factor $\gamma'$, to an instance
$(A,T)$ of the problem in the corollary.
We apply our completion~\cref{thm-main}, using clauses~(\ref{main-clause-1}), (\ref{main-clause-3}), and~(\ref{main-clause-4}),
with $H=2m$ to efficiently construct
 a self-dual $A$ where
$d(A)=2d(C)$ and $T=2t$.
In the YES case, $d(A)\leq T$.
In the NO case, $d(A)\geq\gamma'T>\gamma T$.
Finally, our reduction, together with Bhattiprolu et al.'s
proof, remains efficient and PCP-free.
\end{proof}

%%%%%%%%%%%%%%%%%%%%%%%%%%%%%

\paragraph{Remark:} The above corollaries rely on preserving distance, whereas a different
direction minimizes the number of coordinates added to achieve
self-orthogonality.
An, Kaplan, Kim, Luo, and
\linebreak Wang~\cite{AnEtAl25} study minimizing the number
of columns that must be added to a generator matrix to make its code
self-orthogonal while retaining its input dimension. Their
embeddings (of Hamming codes of length $2^r-1$ with $r\geq3$) are self-dual.
However, their result does not guarantee self-duality for arbitrary input
codes and does not preserve the source distance gap.
In contrast, our metric completion allows the dimension to grow and guarantees self-duality.

\paragraph{Type I and Type II generalizations.}

\Cref{typed-realizations} extends clauses~(\ref{main-clause-1})-(\ref{main-clause-4}) of our main theorem to Type I and Type II completions
with distance scaling factors of two and four, respectively. These extensions also provide linear recovery maps. We stress that these generalizations are not used in the main body of the paper.

%%%%%%%%%%%%%%%%%%%%%%%%%%%%%%%%%%%%
\paragraph{Organization.}
\Cref{preliminaries} sets up the notation and presents self-contained linear-algebra
descriptions of quantum codes and graph distance.
We also present a code-to-graph reduction that we use in the graph-state application in a later section.
\Cref{completion} proves clauses~(\ref{main-clause-1})-(\ref{main-clause-4}) of our main theorem: it explains the construction and proof of the metric
self-dual completion theorem and the companion guarantee for $K$ subspace, which will be helpful in applications.
\Cref{applications} gives the hardness results for both quantum distance and graph-distance applications of clauses~(\ref{main-clause-1}), (\ref{main-clause-3}), and~(\ref{main-clause-4}) of our main theorem.
\Cref{hosts} builds the host construction, which is one of the input ingredients needed for clauses~(\ref{main-clause-2})-(\ref{main-clause-4}) of our main theorem.
\Cref{typed-realizations} extends clauses~(\ref{main-clause-1})-(\ref{main-clause-4}) of our main theorem to
Type I and Type II codes with a linear recovery map.

%%%%%%%%%%%%%%%%%%%%%%%%%%%%%%%%%%%%%%%%%
%%%%%%%%%%%%%%  SECTION 2 %%%%%%%%%%%%%%%%%%
%%%%%%%%%%%%%%%%%%%%%%%%%%%%%%%%%%%%%%%%%%%%%

\section{Preliminaries}
\label{preliminaries}

\subsection{Notation and conventions}
\label{ss:notation}

All lengths, dimensions, indices, repetition parameters, weight thresholds,
and Hamming-ball radii are integers.
Vector spaces are over $\F$. $\wt(x)$ is the Hamming weight, and $\supp(x)$ is the support of $x$. $U\leq V$ means that $U$ is a linear subspace of $V$.
For $U,W\leq V$ when $U\cap W=\{0\}$, define {\em direct sum} as
$U\oplus W\defeq\{u+w:u\in U,\ w\in W\}$. $C^\perp$ is the orthogonal complement of $C$.
A code $C$ is \emph{self-orthogonal} if $C\leq C^\perp$ and
\emph{self-dual} if $C=C^\perp$. Recall that the distance of a code is
$d(C)=\min_{0\neq c\in C}\wt(c)$, where $d(0)=+\infty$.
As usual,
$\dist(y,C)=\min_{c\in C}\wt(y+c)$. The quotient metric is
$\dist(x+C,y+C)\defeq\dist(x+y,C)$.
Self-dual codes are \emph{Type II} if all weights are
divisible by four; they are \emph{Type I} otherwise, and, of course, of even length. We will use
\begin{equation}
 \wt(x+y)=\wt(x)+\wt(y)-2|\supp(x)\cap\supp(y)|
 \label{equation-weight-overlap}
\end{equation}
Codes are specified by bases. Generator matrices have
basis rows. Vectors are columns. Matrices also represent linear maps.
All our algorithms are deterministic.

%%%%%%%%%%%%%%%%%%%%%%%%%%%%%%%%%%%%%%%%%%%%%%%%%%%%
\subsection{A linear algebra view of quantum codes}
\label{quantum-preliminaries}

For  $V=\F^N\times\F^N$,  \emph{symplectic weight} is defined as
$\swt(x,z)\defeq|\supp(x)\cup\supp(z)|$ where $x,z\in\F^N.$ The {\em symplectic pairing} (also called {\em symplectic form}) is defined as
$\omega((x,z),(x',z'))\defeq x^\top z'+z^\top x'$.
Observe that it is nondegenerate and bilinear. For $S \leq V$ {\em symplectic orthogonal complement} is defined as $S^{\perp_{\mathrm s}}\defeq
\{v\in V :  \forall s\in S, \; \omega(v,s)=0\}$. $S$ is {\em isotropic}
if $S\leq S^{\perp_{\mathrm s}}$.
Nondegeneracy gives $\dim S^{\perp_{\mathrm s}}=2N-\dim S$. An isotropic subspace $S$ is called a {\em stabilizer space}
\cite[Sections~I and II-B1]{KapshikarKundu23} and jointly with $N$ describes, for distance calculations, a {\em stabilizer code} $Q$ with parameter $k=N-\dim S$ (where $k$ is the number of qubits that $Q$ encodes). When $k > 0$ the {\em logical distance} of $Q$ is defined as:
\begin{equation}
 d(Q)\defeq\min_{v\in S^{\perp_{\mathrm s}}\setminus S}\swt(v)
 \label{equation-logical-distance}
\end{equation}
When $K\leq \F^N$ is
 self-orthogonal, observe that $S = K \times K$ is isotropic.
The quantum code is denoted as $Q=\operatorname{CSS}(K,K)$, named after
Calderbank and Shor~\cite{CalderbankShor96} and Steane~\cite{Steane96}.
Observe that $S^{\perp_{\mathrm s}}=K^\perp\times K^\perp$,
$k=N-2\dim K$, and, for $k>0$,
\begin{equation}
 d(Q)=\min_{u\in K^\perp\setminus K}\wt(u)
 \label{equation-CSS-distance}
\end{equation}
Since $S^{\perp_{\mathrm s}}=K^\perp\times K^\perp$ and $S=K\times K$,
every $(x,z)\in S^{\perp_{\mathrm s}}\setminus S$ has at least one
component in $K^\perp\setminus K$.
Thus, for $(x,z)\in S^{\perp_{\mathrm s}}\setminus S$,
$\swt(x,z)\geq\max\{\wt(x),\wt(z)\}\geq
\min_{u\in K^\perp\setminus K}\wt(u)$,
with equality at $(u,0)$ for a minimum-weight
$u\in K^\perp\setminus K$.
The minimum symplectic weight of a nonzero vector in $S=K\times K$,
called the \emph{minimum nonzero stabilizer weight}, is defined as:
\begin{equation}
 d(K)=\min_{0\neq s\in S}\swt(s)
 \label{equation-min-stabilizer}
\end{equation}
with the minimum equal to $+\infty$ when $K=0$.

%%%%%%%%%%%%%%%%%%%%%%%%%%%%%%%%%%%%%%%%%%%%%%%%
\subsection{A linear algebra view of graph distance}
\label{graph-preliminaries}
Let $G$ be a simple unweighted graph on $N$ vertices. We denote by $\Gamma$ its adjacency matrix. The notion of {\em graph-state stabilizer space} is defined as
$S\defeq\{(x,\Gamma x):x\in\F^N\}$. Since $G$ is a simple, $\Gamma=\Gamma^\top$, which makes $S$
isotropic, with $\dim S=N$. Thus, $S=S^{\perp_{\mathrm s}}$. That means that
 $k=0$ (hence, no qubits are encoded). The distance of $G$ is defined as the smallest
symplectic weight of a nonzero stabilizer label
\cite[Section~1.2]{GJS25}:

\begin{equation}
 d(\Gamma)\defeq\min_{x\in\F^N\setminus\{0\}}|\supp(x)\cup\supp(\Gamma x)|
 \label{equation-graph-distance}
\end{equation}
If $x$ is a selector function for a vertex set $X$, then $\supp(\Gamma x)$ consists
of vertices with an odd number of neighbors in $X$, giving an equivalent
combinatorial definition.
Graph distance minimizes over $S\setminus\{0\}$, whereas logical
distance in \cref{equation-logical-distance} minimizes over
$S^{\perp_{\mathrm s}}\setminus S$ for a code $Q$ with $k>0$, i.e., encoding at least one qubit.
\begin{lemma}[Grigorescu, Jha, and Samperton,
{\cite[proof of Theorem~2]{GJS25}}]
\label[lemma]{self-dual-graph-distance}
Let $A\leq\F^N$ with dimension $0<r<N$.  After a coordinate permutation, take the parity-check matrix of $A$ and
put it into systematic form $[I_{N-r}\mid R]$.
A simple bipartite graph with adjacency matrix
$$
 \Gamma=\begin{pmatrix}0&R\\R^\top&0\end{pmatrix}
$$
has $N$ vertices,  and
$d(\Gamma)=\min\{d(A),d(A^\perp)\}$.
\end{lemma}
If  $A$ is self-dual, this graph is balanced and has $d(\Gamma)=d(A)$.
Moreover, $R$ is square, and self-duality gives
$RR^\top=R^\top R=I$, so $\Gamma^2=I$.

%%%%%%%%%%%%%%%%%%%%%%%%%%%%%%%%%%%%%%%%%
%%%%%%%%%%%%%%  SECTION 3 %%%%%%%%%%%%%%%%%%
%%%%%%%%%%%%%%%%%%%%%%%%%%%%%%%%%%%%%%%%%%%%%

\section{Metric self-dual completion}
\label{completion}

In this section, we prove clauses~(\ref{main-clause-1})-(\ref{main-clause-4}) of our main result and its companion guarantees. The only ingredient whose proof we delay to \cref{hosts} is the
following  proposition:

\begin{proposition}[Linear-size self-dual code]
\label[proposition]{explicit-self-dual-code}
For every $\Delta\geq1$, there is a deterministically and efficiently constructible self-dual
$P\leq\F^N$ such that $4\mid N$, $d(P)\geq\Delta$, and
$N=\Theta(\Delta)$.
\end{proposition}

\noindent
\Cref{hosts} proves this proposition starting from
the dual-good construction in Shpilka's paper, whose Appendix~A is by Guruswami~\cite{Shpilka09}. In that section, we show how to go from a dual-good code to a self-dual code without losing constant relative rate and constant relative distance.  We now proceed with our main construction.

\subsection{Main construction: building metric self-dual}
\label{structural-guarantees}

We start with an arbitrary input length $m$ and a threshold $H\in\mathbb{Z}$ with $H\geq2m$.
We use \cref{explicit-self-dual-code} with $\Delta=2H+1$ to
choose a self-dual \emph{host} $M\leq\F^N$ with $d(M)>2H$ and
$N=\Theta(H)$. The host and its length depend only on $H$.
$N\geq d(M)>2H\geq2m$.
We define $\tau:\F^m\hookrightarrow\F^N$ as  follows:
$$\tau(x)\defeq(x,x,0^{N-2m})$$

It is immediate that the image $W=\tau(\F^m)$ is self-orthogonal
(since $\tau(x)^\top\tau(y)=x^\top y+x^\top y=0$ for all $x,y\in\F^m$).
Moreover, for every $x\in\F^m$,
\begin{equation}
 \wt(\tau(x))=2\wt(x)
 \label{equation-main-weight}
\end{equation}
Therefore, every word of $W$ has weight at most
 $2m\leq H$.
Since $d(M)>2H$ and $\wt(\tau(x))\leq H$, the triangle inequality gives,  $\forall$ $u\in M\setminus\{0\}$, and $\forall$ $x\in\F^m$,
\begin{equation}
\label{equation-main-separation}
\wt(u+\tau(x))>H
\end{equation}

%%%%%%%%%%%%%%%%%%%%%%%%%%%%%%%%%%%%%%%%%%%%%%%%

We call the host $M$ \emph{protected through $H$} when
\cref{equation-main-separation} holds.  Furthermore, observe that $M\cap W=0$.

As a second step, starting with any  $C\leq\F^m$, we want to insert $\tau(C)$
into host $M$ while preserving self-duality.  Since we need to keep the subspace self-orthogonal, we will keep only host words that are orthogonal to $\tau(C)$. So define
$$K\defeq M\cap\tau(C)^\perp$$

As a third step, we add $\dim(\tau(C))$ dimensions to $K$ to obtain:
$$A\defeq K+\tau(C)$$

$K$ is self-orthogonal since $M$ is, and $K\leq M$.  Furthermore, $\tau(C)$ is self-orthogonal by our construction.
$K$ and $\tau(C)$ are
 orthogonal to each other since $K \leq \tau(C)^\perp$. Since
 $K\cap\tau(C)\subseteq M\cap W=\{0\}$, we have $A=K\oplus\tau(C)$. Therefore,  $A$ is self-orthogonal.

To prove self-duality, the only missing step is to show that when we go from
$M$ to $K$, the operation removes from $M$ exactly $\dim C$ dimensions to get $K$. When we build $A$ we add $\tau(C)$ dimensions back to $K$ and therefore $A$ has
dimension $N/2$. The same calculation also identifies the short
vectors in $K^\perp\setminus K$, which will later determine the quantum distance.

%%%%%%%%%%%%%%%%%%%%%%%%%%%%%%%%%%

\paragraph{Unique decomposition.}

Using the fact that $K\leq M$, $M\cap W=0$, and $\tau(C)\leq W$, we obtain
the following equations:
\begin{equation}
 A=K\oplus\tau(C)
 \label{equation-main-code-sum}
\end{equation}
and
\begin{equation}
 A+W=K\oplus W
 \label{equation-main-direct-sums}
\end{equation}
Therefore, every word of $A+W$ can be written as
$u+\tau(x)$, with  $u\in K$ and $x\in\F^m$.
Observe that $M$, length $N=\Theta(H)$, and embedding
$\tau$  only depends on $m,H$, so the same choices work for any code of length $m$.
Furthermore, the host $M$ is computable in deterministic poly-time in $N$ as shown in
\cref{explicit-self-dual-code}. The rest is just Gaussian elimination and direct-sum decomposition. Therefore, we can compute in poly($N$) bases for
$W,K,A,A+W$.

%%%%%%%%%%%%%%%%%%%%%%%%%%%%%%%%%%%%%%%%%%%%%%

\begin{lemma}[Self-dual completion]
\label[lemma]{self-dual-exchange}
For $K=M\cap\tau(C)^\perp$ and $A=K+\tau(C)$ constructed above,
$\dim K=N/2-\dim C$, and $A$ is self-dual.
Moreover, $K$ is self-orthogonal, $d(K)>2H$, and
\begin{equation}
 \{v\in K^\perp\setminus K:\wt(v)\leq H\}
 =\tau(C\setminus\{0\})
 \label{equation-completion-dual-short-words}
\end{equation}
Consequently, if $C\neq0$, then
\begin{equation}
 \min_{v\in K^\perp\setminus K}\wt(v)=2d(C)
 \label{equation-completion-dual-distance}
\end{equation}
\end{lemma}

\begin{proof}
We have already established in \cref{structural-guarantees} that $A$ is self-orthogonal and
$A=K\oplus\tau(C)$. Since $\tau$ is injective, we have $\dim C= \dim \tau(C)$ and therefore we only need to show that
$\dim K=N/2-\dim C$.
Self-duality of $M$ gives us
$$
K=M^\perp\cap\tau(C)^\perp=(M+\tau(C))^\perp
$$
But we also know that $M\cap W=0$ and $\tau(C)\leq W$. Therefore
$M+\tau(C)$ is direct.   That is, $K=(M\oplus \tau(C))^\perp$. Therefore
\begin{equation}
 K^\perp=M\oplus\tau(C)
 \label{equation-main-orthogonal-complement}
\end{equation}

%%%%%%%%%%%%%%%%%%%%%%%%%%%%%%%%%%%%%%%%%%%%%%%%%%
\noindent
Since $\dim K = N- \dim K^\perp$, it follows that
$\dim K=N-\dim(M+\tau(C))=N/2-\dim C$. By
\cref{equation-main-code-sum},
$\dim A=\dim K+\dim C=N/2$. Finally, since $A$ is self-orthogonal, $A=A^\perp$.
Since $K\leq M$, is self-orthogonal and $d(K)\geq d(M)>2H$, by \cref{equation-main-orthogonal-complement}, every word of $K^\perp$
is uniquely defined as $u+\tau(c)$ with $u\in M$ and $c\in C$.
If $u\neq0$, its weight exceeds $H$ by
\cref{equation-main-separation}. If $u=0$, its weight is
$2\wt(c)\leq H$, and it lies outside $K$ exactly when $c\neq0$,
because $K\cap\tau(C)=0$.
This proves \cref{equation-completion-dual-short-words}.
For $C\neq0$, these short words have minimum weight $2d(C)$,
which gives \cref{equation-completion-dual-distance}.
\end{proof}

%%%%%%%%%%%%%%%%%%%%%%%%%%%%%%%%%%%%%%%%%%%%%%%%%%%%
\subsection{Preservation of source-coset distances}
\label{coset-guarantees}

For any  $y\in\F^m$,  let's consider the coset $\tau(y)+A$,  and its relation to the original coset $y+C$.
Since we have already shown that $A=K\oplus\tau(C)$, every word of $\tau(y)+A$ has a unique
representation that can be written as $u+\tau(x)$, with $u\in K$ and $x\in y+C$.
When $u=0$, the word is $\tau(x)\in\tau(y+C)$, with weight
$2\wt(x)\leq H$. On the other hand, if $u\neq0$, its weight exceeds
$H$ by \cref{equation-main-separation}. Therefore, a minimum weight
representative of coset  $\tau(y)+A$ must have $u=0$.
In other words, every minimum-weight representative of $\tau(y)+A$
has zero host component $u=0$ and weight exactly $2\dist(y,C)$.
This proves
\begin{equation}
 \dist(\tau(y),A)=2\dist(y,C)
 \label{equation-coset-metric}
\end{equation}
The same argument gives
\begin{equation}
 \{p\in\tau(y)+A:\wt(p)\leq H\}=\tau(y+C)
 \label{equation-coset-cutoff}
\end{equation}
The direct sum equations in
\cref{equation-main-code-sum,equation-main-direct-sums} show that $\tau$ induce an isomorphism
$\F^m/C\to(A+W)/A$.
Applying this to $y+z$ shows that this isomorphism
exactly doubles the distance between any two source cosets
$y+C$ and $z+C$.

\begin{proof}[Proof of \Cref{thm-main}.]
The construction in \cref{structural-guarantees} picks
$M,N,\tau$ independently of $C$, with $d(M)>2H$ and $N=\Theta(H)$.
By \cref{self-dual-exchange}, we know that $A=A^\perp$.
Now, \cref{equation-main-weight,equation-coset-metric,equation-coset-cutoff}
give clauses~(\ref{main-clause-1})-(\ref{main-clause-3}) of the main theorem, where we set $y=0$ in
\cref{equation-coset-cutoff}.
The coset correspondence and our pairwise distance guarantees
are already proved in \cref{coset-guarantees}.
For $C\neq\{0\}$, \cref{equation-main-weight,equation-coset-cutoff} (with $y$ set to $0$) give
\begin{equation}
 d(A)=2d(C)\leq H
 \label{equation-completion-distance}
\end{equation}
For $C=\{0\}$, our construction already sets $A=M$, so $d(A)>2H$.
The running time in clause~(\ref{main-clause-4}) is analyzed in  \cref{structural-guarantees}.
\end{proof}

%%%%%%%%%%%%%%%%%%%%%%%%%%%%%%%%%%%%%%%%%%%%%%%%%%%%%%%%%%%%%%
%%%%%%%%%%%%%%%%%%%%%%%%% SECTION 4 %%%%%%%%%%%%%%%%%%%%%%
%%%%%%%%%%%%%%%%%%%%%%%%%%%%%%%%%%%%%%%%%%%%%%%%%%%%%%%%%

\section{Hardness applications}
\label{applications}

Our two additive hardness gaps use different outputs of our main construction.
For quantum minimum distance, we use $K$ for CSS check spaces.
We use a self-dual code $A$ as the input to the code-to-graph reduction
of Grigorescu, Jha, and Samperton~\cite[proof of Theorem~2]{GJS25}.
Both hardness constructions exactly double the classical source distance and retain the
linear output length.
The Austrin-Khot hardness theorem supplies us with the hardness gap to get a
linear additive gap for both settings.

%%%%%%%%%%%%%%%%%%%%%%%%%%%%%%%
\subsection{Application to Quantum code distance}
\label{quantum}

\begin{lemma}[Identical-check Minimum Distance]
\label[lemma]{mindist-CSS}
For all nonzero  $C\leq\F^m$ and $H\geq2m$,
in poly$(m+H)$ time we can construct
$K\leq\F^N$ such that
$K$ is self-orthogonal and $Q=\operatorname{CSS}(K,K)$  where $N=\Theta(H)$,
$k=2\dim C$, $d(Q)=2d(C)$, and $d(K)>2H$.
\end{lemma}

\begin{proof}
Start with  $K$ from  \cref{self-dual-exchange} and let
$Q=\operatorname{CSS}(K,K)$. Our construction guarantees that
$K$ is self-orthogonal, $d(K)>2H$, and
$\dim K=N/2-\dim C$.
\cref{quantum-preliminaries} shows that
$k=N-2\dim K$ and $k=2\dim C$.
Therefore, using \cref{equation-CSS-distance,equation-completion-dual-distance},
gives us $d(Q)=2d(C)$.
The reduction is deterministic and efficient \cref{structural-guarantees}.
\end{proof}

%%%%%%%%%%%%%%%%%%%%%%%%%%%%%%%%%%%%%%%%%%%%%%%
\begin{proof}[Proof of \Cref{thm-intro-quantum-hardness}]
Pick a constant $\lambda>1$ and start with an Austrin-Khot hard
instance $(C,t)$, where $C\leq\F^m$, from
\cref{classical-consequences}.
They show that, for some constant $\gamma>1$, it is NP-hard to distinguish
$d(C)\leq t$ from $d(C)>\gamma t$, even when the dimension and distance
of $C$ are linear in $m$ (i.e., $\dim C=\Omega(m)$ and
$d(C)=\Omega(m)$).
We now apply \cref{mindist-CSS} to $C$ with
$H=2\lceil\lambda\rceil m$. Observe that
$Q=\operatorname{CSS}(K,K)$ acts on
$N=\Theta(m)$ qubits.
Moreover,
$k=2\cdot\dim C=\Omega(N)$ and
$d(Q)=2\cdot d(C)=\Omega(N)$.
Set $T=2t$. In the YES case, $d(Q)\leq T$. In the NO case,
$d(Q)>\gamma T$, and hence
$$
d(Q)-T>
\left(1-\frac1\gamma\right)d(Q)
=\Omega(N)
$$
So we can find $c,r,\delta>0$ so that
$k\geq rN$ and $d(Q)\geq\delta N$, with an additive hardness gap of $cN$.
Moreover, every nontrivial stabilizer has weight $\geq d(K)$.
Since $C\leq\F^m$, its distance is at most $m$. Since we set
$H=2\lceil\lambda\rceil m$ and have $d(Q)=2d(C)$, this guarantees that
$d(K)>H\geq\lambda\cdot d(Q)$.
Construction of $Q$ is poly$(N)$-time and deterministic.
\end{proof}

%%%%%%%%%%%%%%%%%%%%%%%%%%%%%%%%%%%%%%%%%%%%%%%%%%%%%%%%%%%%%%%%%%%%
%%%%%%%%%%%%
\subsection{Application to Graph-state distance}
\label{graphs}

By \cref{self-dual-graph-distance}, a self-dual code and its graph
have the same distance. Combining this fact with our completion
gives the following:

\begin{theorem}[Exact reduction to graph-state distance]
\label{graph-exact-reduction}
For nonzero $C\leq\F^m$, one can construct in
poly$(m)$ time a simple balanced bipartite graph on
$N=\Theta(m)$ vertices with adjacency matrix $\Gamma$ that satisfies
$\Gamma^2=I$ and $d(\Gamma)=2d(C)$.
\end{theorem}

\begin{proof}
As in the previous case, apply \cref{thm-main}, using clauses~(\ref{main-clause-1}), (\ref{main-clause-3}), and~(\ref{main-clause-4}), with $H=2m$.
We get a self-dual $A$
of length $N=\Theta(m)$, and by
\cref{equation-completion-distance} $d(A)=2d(C)$.
\cref{self-dual-graph-distance} shows how to build a simple balanced
bipartite graph $G$ on $N$ vertices with adjacency matrix $\Gamma$
where $\Gamma^2=I$ and $d(\Gamma)=d(A)=2d(C)$.
Both constructions are deterministic and efficient.
\end{proof}

\begin{proof}[Proof of \Cref{graph-gap}]
We start with a hard self-dual instance $(A,T)$ of length $N$
from \Cref{self-dual-distance-hardness} and apply
\cref{self-dual-graph-distance}.
The reduction $(A,T)\mapsto(\Gamma,T)$ preserves both distance and length,
is efficient and deterministic. The resulting simple balanced bipartite
graph satisfies $\Gamma^2=I$.
Thus, the reduction preserves the linear additive gap and a constant
relative distance.
Finally, observe that if $c_0N$ is the source gap, an estimate within additive error
$c_0N/2$ distinguishes the YES and NO instances by comparison
with $T+c_0N/2$. Now, set $c=c_0/2$.
\end{proof}

%%%%%%%%%%%%%%%%%%%%%%%%%%%%%%%%%%%%%%%%%%%%%%%%
%%%%%%%%%%%%%%%%%% SECTION 5 %%%%%%%%%%%%%%%%%%%
%%%%%%%%%%%%%%%%%%%%%%%%%%%%%%%%%%%%%%%%%%%%%%%%%

\section{Linear-size self-dual hosts}
\label{hosts}

We now prove
\cref{explicit-self-dual-code}, which provides the host $M$ that we use in
\cref{completion}. We stress that this proof is independent of the completion and hardness arguments presented in the previous sections.
We start with a theorem of Shpilka and Guruswami which gives an explicit ``dual-good'' code with constant rate and constant
relative distance, whose dual also has a constant rate and
constant relative distance.
Let us first restate Theorem~4 in Shpilka's paper in our notation.
The construction and proof are given in Appendix~A of Shpilka's
paper, which states that the appendix was written by Guruswami:
\begin{theorem}[Shpilka-Guruswami, {\cite[Theorem~4]{Shpilka09}}]
\label{shpilka-guruswami-dual-good}
$\exists$ constant $\zeta\geq1/30$ such that, $\forall$
$j\geq1$, $\exists$ code $D_j\leq\F^{n_j}$
with $n_j=42\cdot8^{j+1}$, with $\dim D_j=n_j/2$, and
with $\min\{d(D_j),d(D_j^\perp)\}\geq\zeta n_j$.
Construction, encoding, and decoding from a fraction $\zeta/2$
of errors all take poly$(n_j)$ time and are deterministic.
\end{theorem}

The statement in Shpilka's paper denotes the two relative distances by
$\eta,\eta'\geq\zeta$.
For our purposes, set $\delta=1/30$ and observe that distances are at least $\delta n_j$, and
$n_{j+1}=8n_j$.
We compute a basis of $D_j$ by encoding the standard basis
vectors of $\F^{n_j/2}$.
We do not need the
decoding algorithm, and we do not make an assumption that
$D_j$ contains the all-one word.
Instead, we take the Shpilka-Guruswami theorem as a black box theorem and modify
any such dual-good code in two modular steps using elementary linear algebra:

\begin{enumerate}
\item We first show how to put an all-ones word into any dual-good code without losing its dual-good property.
The idea is to append the code and its dual, then apply
an invertible linear map to each coordinate pair.
\item Next, we show how, starting with any dual-good code containing an all-ones vector, we can construct a self-dual code with a constant rate and constant relative distance. The idea is to use prefix sums to obtain self-duality.
\end{enumerate}

Each step of this preliminary construction preserves a lower bound on
distances proportional to the original length and doubles the length of the code.
We denote $\ones_n$ for the all-ones vector in $\F^n$.

%%%%%%%%%%%%%%%%%%%%%%%%%%%%%%%%%%%%%%%%%%%%%%%%%%%%%%%%
\subsection{Putting the all-ones word into any dual-good code.}

The following normalization preserves the smaller of the code and
dual distances.  Swapping the two coordinate blocks of the constructed
code gives its dual.

\begin{lemma}[All-ones normalization]
\label[lemma]{all-ones-symmetrization}
For every code $D\leq\F^n$, where $n\geq1$, one can efficiently construct
a code $E\leq\F^{2n}$ such that
$\dim E=n$, $\ones_{2n}\in E$, and
$$\min\{d(E),d(E^\perp)\}\geq\min\{d(D),d(D^\perp)\}$$
\end{lemma}
\begin{proof}
For every $h\in D\cap D^\perp$, we have
$\ones_n^\top h=h^\top h=0$ (in $\F$). Thus,
$\ones_n\in(D\cap D^\perp)^\perp=(D+D^\perp)$.
We solve for $a\in D$ and $b\in D^\perp$ with $a+b=\ones_n$, and define $E$ where multiplication is coordinatewise
as:
$$
 E=\{(ax+by,\,x+y):x\in D,\ y\in D^\perp\}
$$
Observe that from any output $(r,s)=(ax+by,x+y)$, we can recover the input $x=r+bs$ and $y=r+as$.
Thus, the map is linear and injective.  Therefore
$\dim E=n=\dim D+\dim D^\perp$.
If we set $(x,y)=(a,b)$ this gives $(r,s)=(\ones_n,\ones_n)$. Therefore
$\ones_{2n}\in E$. Now, we can show that swapping the two blocks of $E$ gives its dual: consider any two
output pairs $(r,s)$ and $(r',s')$ that came from two input pairs $(x,y)$ and $(x',y')$.
Expanding and using $a+b=\ones_n$ gives:
$$
 r^\top s'+s^\top r'=x^\top y'+y^\top x' = 0 + 0
$$
Therefore, the code obtained when we swap the two blocks of $E$ is contained in
 $E^\perp$.  Since both have dimension $n$,
$$
 E^\perp=\{(s,r):(r,s)\in E\}
$$
and
 $d(E^\perp)=d(E)$. Now, observe that
$$
 \wt(r,s)=\wt(r)+\wt(s)\geq\max\{\wt(x),\wt(y)\}
$$
This follows from $x=r+bs$ and $y=r+as$.
When $(r,s)\ne0$, at least one of $x,y$ is nonzero; therefore,
$\wt(r,s)\geq\min\{d(D),d(D^\perp)\}$.  This means that both
 $E$ and $E^\perp$ have a distance of at least $\min\{d(D),d(D^\perp)\}$.
\end{proof}

%%%%%%%%%%%%%%%%%%%%%%%%%%%%%%%%%%%%%%%%%%%%%%
\subsection{The prefix-sum construction.}

We next turn a code $E$ containing the all-ones word into a
self-dual code. We repeat each coordinate of a word in $E$ and
add consecutive prefix sums of a word in $E^\perp$.
Summing each output pair recovers the latter word.
Let $E\leq\F^n$ and assume that $\ones_n\in E$. For $w\in E^\perp$, we define a prefix sum of all bits of $w$:
$$s_0 \defeq  s_0(w)=0$$
and for $1\leq i\leq n$:
$$s_i \defeq s_i(w)=\sum_{j=1}^i w_j$$
Observe that $s_0=0$ by definition. Since $w\in E^\perp$, it  is orthogonal to $\ones_n$, and thus $s_n(w)=0$. Define $\Phi:E\times E^\perp\to\F^{2n}$ by
$$\Phi(c,w) \defeq \left(\Phi_1(c,w),\ldots,\Phi_n(c,w)\right)$$
where
$$
 \Phi_i(c,w)=\bigl(c_i+s_{i-1}(w),\,c_i+s_i(w)\bigr)
$$
The output consists of $n$ pairs, iterating over coordinates. %
Let $P=\operatorname{im}\Phi$.  Our new code is
$$P=\Phi(E,E^\perp)$$

\begin{lemma}[Prefix-sum self-dual construction]
\label[lemma]{prefix-sum-construction}
$\forall$ codes $E\leq\F^n$ s.t. $\ones_n\in E$,  $\forall$ $n\geq1$,
the map $\Phi:E\times E^\perp\to P$ is an invertible linear map. $\forall$
$c,c'\in E$ and $w,w'\in E^\perp$,
$$\Phi(c,w)^\top\Phi(c',w')=c^\top w'+w^\top c'$$
Moreover, code $P$ is self-dual and has distance \;
$d(P)\geq\min\{d(E^\perp),\,2d(E)\}$.
\end{lemma}

\begin{proof}
$ $

\noindent
\textbf{Linearity and dimension:} $\Phi$ is linear because prefix sums and coordinate addition are linear. To prove that it is one-to-one, consider an output
$z=\Phi(c,w)$. From $z$, we can use the prefix-sum structure to recover $w$ and then recover $c$: $w_i=z_{2i-1}+z_{2i}$ and $c_i=z_{2i-1}+s_{i-1}(w)$.
Thus,  $\dim P=\dim E+\dim E^\perp$ which is $n$.

%%%%%%%%%%%%%%%%%%%%%%%%%%%%%%%%%%%%%%%%%%%%%%%%%%%%%%%%%%%%%%
\medskip
\noindent
\textbf{Self-orthogonality:} Let $c,c'\in E$ and $w,w'\in E^\perp$.
Now, we want to show that it is always the case that $\Phi(c,w)^\top\Phi(c',w')=0$.
Linearity of $\Phi$ allows us to split arguments as:
$$\Phi(c,w)=\Phi(c,0)+\Phi(0,w)$$ where type one is:
$$\Phi(c,0)=(c_1,c_1,\ldots,c_n,c_n)$$
and type two is:
$$\Phi(0,w)=(0,s_1(w),s_1(w),\ldots,s_{n-1}(w),s_{n-1}(w),0)$$

For  $\Phi(c,0)^\top\Phi(c',0)$  each coordinatewise product $c_ic'_i$ occurs
twice, and therefore cancels over $\F$.
For $\Phi(0,w)^\top\Phi(0,w')$ , each coordinatewise product
$s_i(w)s_i(w')$, for $1\leq i<n$, also occurs twice, in positions $(2i,2i+1)$, and the first and last points contribute zero, which results in cancellation as well.

Now, let's consider cross-products.
For $\Phi(c,0)^\top\Phi(0,w')$,  output pair $i$ contributes
$$c_i(s_{i-1}(w')+s_i(w'))=c_iw'_i$$ Hence
$$
 \Phi(c,0)^\top\Phi(0,w')=c^\top w'=0
$$
Observe that the last equality holds since $c\in E$ and $w'\in E^\perp$.
The same argument applies to the other cross-product pair $\Phi(0,w)^\top\Phi(c',0)=w^\top c'$.
Finally, expanding the dot product of two complete outputs gives four terms, where
$$
 \Phi(c,w)^\top\Phi(c',w')=c^\top w'+w^\top c'=0
$$
Thus $P\subseteq P^\perp$. Finally, since we already showed that $\dim P=n$ and
the length of $P$ is $2n$, we have $\dim P^\perp=2n-n=n$, so $P=P^\perp$.
%%%%%%%%%%%%%%%%%%%%%%%%%%%%%%%%%%%%%%%%%%%%%%%%%%

\medskip
\noindent
\textbf{Distance.} Take a nonzero output $z=\Phi(c,w)$.
If $w\ne0$, for coordinates where  $w_i=1$,  pair $i$ has the form:
$$(c_i+s_{i-1}(w),c_i+s_{i-1}(w)+1)$$
which forces this pair to be either $(1,0)$ or $(0,1)$ and therefore
$$
 \wt(z)\geq\wt(w)\geq d(E^\perp)
$$
On the other hand, if $w=0$, every prefix sum is zero, so $z$ repeats each
coordinate of $c$.  We know that $z\ne0$, we have $c\ne0$, and $c\in E$ contributes to the weight:
$$
 \wt(z)=2\wt(c)\geq2d(E)
$$
These two cases prove the distance bound.

\medskip
\noindent
\textbf{Construction time.} Gaussian elimination computes a basis of
$E^\perp$ from a basis of $E$. Applying $\Phi$ to the pairs
$(u,0)$ and $(0,v)$, where $u$ and $v$ range over these two bases,
produces a basis of $P$ by injectivity. All algorithms are deterministic and polynomial-time.
\end{proof}

%%%%%%%%%%%%%%%%%%%%%%%%%%%%%%%%%%%%%%%%%%%%%%%%%%%%%%%%%%%
\subsection{Constructing linear-size self-dual codes}

We can now go directly from dual-good to self-dual codes:
\begin{proposition}[From dual-good to self-dual codes]
\label[proposition]{dual-good-to-self-dual}
Let $\delta>0$.
Given $D\leq\F^n$, with both $d(D)\geq\delta n$, and $d(D^\perp)\geq\delta n$.
We can deterministically and in poly($n$)-time construct a self-dual code $P\leq\F^{4n}$ with $\dim P=2n$ and $d(P)\geq\delta n$.
Thus $P$ has rate $1/2$ and relative distance $\geq\delta/4$. \end{proposition}
\begin{proof}
Apply all-ones \cref{all-ones-symmetrization} to $D$ to get $E\leq\F^{2n}$, where $\ones_{2n}\in E$, and
$\min\{d(E),d(E^\perp)\}\geq\min\{d(D),d(D^\perp)\}$.
Then apply prefix-sum
\cref{prefix-sum-construction} to $E$ to get
a self-dual $P\leq\F^{4n}$ with
$d(P)\geq\min\{d(E^\perp),2d(E)\}\geq\min\{d(D),d(D^\perp)\}$.
Since $P$ is self-dual, $\dim P=2n$, which gives us $1/2$ rate.
The distance bound above gives $d(P)\geq\delta n$. Therefore, the relative
distance  $\geq \delta/4$.
\end{proof}

%%%%%%%%%%%%%%%%%%%%%%%%%%%%%%%%%%%%%%%%%%%%%%%%%%%%%%%%%%%%%

\noindent
We now apply \cref{dual-good-to-self-dual} to Shpilka-Guruswami~\cite{Shpilka09} to construct:

\begin{proof}[Proof of \Cref{explicit-self-dual-code}.]
Apply \cref{shpilka-guruswami-dual-good} to construct dual-good code
$D_j\leq\F^{n_j}$ with \linebreak
$\min\{d(D_j),d(D_j^\perp)\}\geq\delta n_j$.
Choose the smallest $j\geq1$ such that $\delta n_j\geq\Delta$.
Shpilka-Gurswami guarantee that $n_j=\Theta(\Delta)$.
(to see this, observe that  $\delta$ and the first available length are constants,
$n_{j+1}=8n_j$,  and for smallest $j$, we have
$n_j=\Theta(\Delta)$).
Now,
apply (to $D_j$) the transformation from dual-good to self-dual codes from  \cref{dual-good-to-self-dual}.
This yields a self-dual $P$ of length $N=4n_j$ with
$d(P)\geq\min\{d(D_j),d(D_j^\perp)\}\geq\delta n_j\geq\Delta$.
Thus,  $4\mid N$ and $N=\Theta(\Delta)$ are as stated.
Both Guruswami's construction of $D_j$ and our conversions are deterministic
and poly($N$) time.
\end{proof}

%%%%%%%%%%%%%%%%%%%%%%%%%%%%%%%%%%%%%%%%%%%%%%%%%%%
%%%%%%%%%%%%%%%%%% ACKS %%%%%%%%%%%%%%%%%%%%%%%%%%%%
%%%%%%%%%%%%%%%%%%%%%%%%%%%%%%%%%%%%%%%%%%%%%%%%%%%
\bigskip
\section*{Acknowledgments}

The mathematical ideas, results, proofs, and exposition in this paper are entirely the author's. After completing the manuscript, the author used OpenAI's ChatGPT only as a final check for possible mathematical inconsistencies and typographical errors. The tool identified a few typographical errors but no errors in the mathematical arguments. The author takes full responsibility for the contents of this paper.

%%%%%%%%%%%%%%%%%%%%%%%%%%%%%%%%%%%%%%%%%%%%%%%%%%%%%
%%%%%%%%%%%%      REFERENCES %%%%%%%%%%%%%%%%%%%%%%%%%%
%%%%%%%%%%%%%%%%%%%%%%%%%%%%%%%%%%%%%%%%%%%%%%%%%%%%%%

\newpage
\bibliographystyle{alpha}
\bibliography{references}

%%%%%%%%%%%%%%%%%%%%%%%%%%%%%%%%%%%%%%%%%%%%%%%%%%%%%%%%%%
%%%%%%%%%%%%%%%%%%%%%%%%  APPENDIX %%%%%%%%%%%%%%%%%%%%%%%
%%%%%%%%%%%%%%%%%%%%%%%%%%%%%%%%%%%%%%%%%%%%%%%%%%%%%%%%%%
%\newpage

\bigskip
\appendix
\titleformat{\section}
  {\normalfont\Large\bfseries}
  {Appendix~\thesection}{1em}{}
\section{Type I and Type II metric self-dual completions}
\label[appendix]{typed-realizations}
We give two generalizations of clauses~(\ref{main-clause-1})-(\ref{main-clause-4}) of our main theorem: to Type I and Type II codes. These generalizations are not used in the main body of the paper, but we present them here for completeness.
 Recall that a code is doubly even if every word weight is divisible by four.
\subsection{The completion and its recovery map}
Let  $m,s,H$  be positive integers, where $H\geq s\cdot m$.
Similar to our \cref{structural-guarantees}, let $M\leq\F^N$ be self-dual
and $\tau:\F^m\to\F^N$ be an injective linear map.
Let $W\defeq\tau(\F^m)$ be self-orthogonal, with
$\wt(\tau(x))=s\cdot\wt(x)$.
Assume that $M$ is protected through $H$, namely,
\begin{equation}
 \wt(u+\tau(x))>H
 \qquad
 \text{for every $u\in M\setminus\{0\}$ and $x\in\F^m$}
 \label{equation-typed-protection}
\end{equation}
The same triangle inequality as in
\cref{equation-main-separation} shows that $d(M)>2H$ is sufficient
for \cref{equation-typed-protection}.
For $C\leq\F^m$, we use the same construction as in our main theorem clauses~(\ref{main-clause-1})-(\ref{main-clause-4}):
$$K\defeq M\cap\tau(C)^\perp$$
and
$$A\defeq K+\tau(C)$$
%%%%%%%%%%%%%%%%%%%%%%%%%%%%%%%
\begin{lemma}[Completion and recovery]
\label[lemma]{metric-exchange}
As in \cref{thm-main}, $A$ is self-dual, and $\tau$ induces
isomorphism
$\F^m/C\to(A+W)/A$
that increases all coset distances by $s$ factor.
The following properties hold simultaneously:
\begin{enumerate}
\renewcommand{\theenumi}{\roman{enumi}}
\renewcommand{\labelenumi}{(\theenumi)}
\item\label{generalized-main-clause-1}
For every $y\in\F^m$ and $0\leq r\leq H$,
\begin{equation}
 \{p\in\tau(y)+A:\wt(p)\leq r\}
 =\tau\bigl(\{x\in y+C:\wt(x)\leq\lfloor r/s\rfloor\}\bigr)
 \label{equation-protected-ball}
\end{equation}
Thus, if $C\neq0$, then $d(A)=s\cdot d(C)$, and if $C=0$, then
$d(A)>H$.
\item\label{generalized-main-clause-2}
Every $p\in A+W=K\oplus W$ has a unique expression
$p=u+\tau(x)$, with $u\in K$ and $x\in\F^m$.
\item\label{generalized-main-clause-3}
The linear ``recovery'' map $\rho:A+W\to\F^m$
\begin{equation}
 \rho(u+\tau(x))\defeq x
 \label{equation-recovery-map}
\end{equation}
guarantees $\rho\circ\tau=\operatorname{id}_{\F^m}$ and
$A=\rho^{-1}(C)$. Note that $\rho$ induces the inverse quotient isomorphism.
For every $p\in A+W$,
$p+\tau(\rho(p))\in A$ and $s\cdot\wt(\rho(p))\leq\wt(p)$.
\item\label{generalized-main-clause-4}
If $M$ and $W$ are doubly even, then $A$ is Type II.
\item\label{generalized-main-clause-5}
If $M\cap W^\perp$ contains a word of weight $2\bmod4$,
then $A$ is Type I.
\item\label{generalized-main-clause-6}
Given $M,\tau,C$, the completion $A$ and recovery map $\rho$ are
computable by Gaussian elimination in poly($N$).
\end{enumerate}
\end{lemma}
%%%%%%%%%%%%%%%%%%%%%%%%%%
\begin{proof}
By \cref{equation-typed-protection}, observe that every nonzero $u\in M$ and every
$x\in\F^m$ enjoy
$\wt(u+\tau(x))>H$.
In particular, $M\cap W=0$. If not the case, then  some nonzero $u\in M\cap W$
could be written as $u=\tau(x)$, which would contradict
$\wt(u+\tau(x))=0$.
Therefore,  \cref{self-dual-exchange} argument
applies unchanged, and
\cref{equation-main-code-sum,equation-main-direct-sums} remains valid.

\emph{Clause~(\ref{generalized-main-clause-2}).}
The direct sum $A+W=K\oplus W$ and $\tau$ injectivity gives the
unique decomposition $p=u+\tau(x)$.

\emph{Clause~(\ref{generalized-main-clause-1}).}
Every word of $\tau(y)+A$ must have this expression with $x\in y+C$.
If $\wt(p)\leq H$, then \cref{equation-typed-protection} forces
$u=0$, and hence $\wt(p)=s\cdot\wt(x)$.
The converse is also immediate, therefore proving
\cref{equation-protected-ball}.
The two direct sums identify $(A+W)/A$ with $\F^m/C$.
Therefore, minimum-weight $x\in y+C$ satisfies
$s\cdot\wt(x)\leq s\cdot m\leq H$, so
\cref{equation-protected-ball} shows that the corresponding coset
has distance which is exactly
$s\cdot\min_{x\in y+C}\wt(x)$.
It also gives $d(A)=s\cdot d(C)$ when $C\neq0$.
If $C=0$, then $A=M$, and
\cref{equation-typed-protection}, (with $x=0$) yields $d(A)>H$.

\emph{Clause~(\ref{generalized-main-clause-3}).}
By clause~(\ref{generalized-main-clause-2}),
\cref{equation-recovery-map} makes $\rho$ linear.
For $p=u+\tau(x)$,
$\rho\circ\tau=\operatorname{id}$, the condition that $p\in A$ holds exactly when $x\in C$,
and $p+\tau(\rho(p))=u\in K\leq A$.
Thus, $\rho$ enjoys all the stated algebraic and quotient properties.
The weight bound holds with equality when $u=0$.
If $u\neq0$, then \cref{equation-typed-protection} yields
$s\cdot\wt(\rho(p))\leq s\cdot m\leq H<\wt(p)$.

\emph{Clause~(\ref{generalized-main-clause-4}).}
Both $K\leq M$ and $\tau(C)\leq W$ are doubly even, and they are
orthogonal.
Hence, \cref{equation-weight-overlap} shows that their sum $A$ is
doubly even.

\emph{Clause~(\ref{generalized-main-clause-5}).}
A word in $M\cap W^\perp$ lies in $K\leq A$.
If its weight is $2\bmod4$, the self-dual code $A$ is therefore
Type I.

\emph{Clause~(\ref{generalized-main-clause-6}).}
Observe that Gaussian elimination computes $K$, $A$, and $\rho$ in poly($N$) time.
\end{proof}

%%%%%%%%%%%%%%%%%%%%%%%%%%%%%%%%%%%%%%%%%%
\subsection{Type I completion}

\label{appendix-type-I}

Recall that puncturing a code in a coordinate means deleting that
coordinate from every codeword. On the other hand, shortening in that
coordinate means restricting only to codewords that are zero in
that coordinate and then deleting this coordinate.

\begin{theorem}[Type I metric self-dual completion]
\label{thm-type-I-completion}
$\forall$  $m,H\in\mathbb{Z}_{>0} $,  $H\geq 2m$, $\exists$ deterministic and efficient way to
construct Type I code $M\leq\F^N$,  $N=\Theta(H)$,
protected through $H$, with one-to-one embedding
$\tau:\F^m\to\F^N$: $\tau(x)=(x,x,0^{N-2m})$.
Our construction depends only on $m$ and $H$ and guarantees that
for every code $C\leq\F^m$, the completion $A=K+\tau(C)$ is Type I
and has all the properties of \cref{metric-exchange} with $s=2$.
\end{theorem}

\begin{proof}
We apply \cref{explicit-self-dual-code} with $\Delta=2H+2$ to obtain
a self-dual code $P\leq\F^n$ with $4\mid n$, $n=\Theta(H)$, and
$d(P)\geq2H+2$. Puncture its last coordinate to form
$E\leq\F^{n-1}$. Since every binary self-dual code contains the
all-ones word, $\ones_{n-1}\in E$ and
$d(E)\geq d(P)-1>2H$. Recall that puncturing and shortening are dual
operations. Since $P=P^\perp$, we obtain
$$
E^\perp=\{w\in\F^{n-1}:(w,0)\in P\}
$$
and $d(E^\perp)\geq d(P)>2H$.

By the prefix-sum construction in \cref{prefix-sum-construction},
$M=\Phi(E,E^\perp)$ is self-dual, has length
$N=2(n-1)=\Theta(H)$, and has distance greater than $2H$.
Observe that since  $n\geq d(P)\geq2H+2$, we have
$N=2(n-1)\geq4H+2\geq2m$.
Thus, for every nonzero $u\in M$ and every $x\in\F^m$,

$$
\wt(u+\tau(x))
\geq d(M)-\wt(\tau(x))
>2H-H=H,
$$

with $\wt(\tau(x))=2\wt(x)\leq2m\leq H$.
Therefore, \cref{equation-typed-protection} holds.
Since $4\mid n$, we have $N=2(n-1)\equiv2\pmod4$.
Every binary self-dual code of length $N$ contains $\ones_N$,
whose weight is $2\bmod4$, and hence is Type I.
Thus, $M$ is Type I.
The embedding doubles weights, and its image $W$ is
self-orthogonal. Every word of $W$ has even weight, so
$\ones_N\in W^\perp$.
Since $\ones_N\in M$ and has weight $2\bmod4$,
clause~(\ref{generalized-main-clause-5}) of
\cref{metric-exchange} applies and therefore makes every completion $A$ Type I.
We now apply clauses~(\ref{generalized-main-clause-1})-
(\ref{generalized-main-clause-3}) of \cref{metric-exchange}.
These give all the recovery and metric properties, while
clause~(\ref{generalized-main-clause-6}) gives efficient maps.
The construction depends only on $m$ and $H$.
\end{proof}

%%%%%%%%%%%%%%%%%%%%%%%%%%%%%%%%%%%%%%%%%%%%%%%
\subsection{Type II completion}
\label{appendix-type-II}

\begin{theorem}[Type II metric self-dual completion]
\label{thm-type-II-completion}
$\forall \; m,H\in\mathbb{Z}_{>0}$ s.t. $4m\leq H$, there is a
deterministic poly-time construction of  Type II self-dual
 $M\leq\F^N$, \; $N=\Theta(H)$, with 1-to-1
embedding
$\tau:\F^m\to\F^N$ defined as
$\tau(x)=(x,x,x,x,0^{N-4m})$, so that $M$ is protected through $H$.
 $M$ and the embedding $\tau$ depend only on $m$ and $H$.
That is,
for every code $C\leq\F^m$, define (as in our main construction)
$
K\defeq M\cap\tau(C)^\perp
$
and
$
A\defeq K+\tau(C)
$.
Then $A$ is a Type II self-dual code and has all the metric and
recovery properties of \cref{metric-exchange} with $s=4$.
\end{theorem}

\begin{proof}
We first apply \cref{explicit-self-dual-code}
with $\Delta=2H+1$ to obtain a self-dual code $P\leq\F^n$
with $n=\Theta(H)$ and $d(P)>2H$.
We then interleave two copies of $P$:
$$
E\defeq\{(a_1,b_1,\ldots,a_n,b_n):a,b\in P\}
$$
Interleaving only permutes the coordinates of $P\oplus P$.
Thus, $E=E^\perp$, $\ones_{2n}\in E$, and $d(E)=d(P)$.
By \cref{prefix-sum-construction}, the code $M=\Phi(E,E)$
is self-dual, has length $N=4n=\Theta(H)$, and has distance
greater than $2H$.

To show that $M$ is Type II, we still need to check that it is
doubly even. Take any $\Phi(c,w)\in M$ and re-write
$\Phi(c,w)=\Phi(c,0)+\Phi(0,w)$.
Observe that $\Phi(c,0)$ has weight $2\wt(c)$.
As $E$ is self-dual, all its codewords have even weight,
so $2\wt(c)$ is divisible by $4$.

Now consider the second summand.
For $w=(a_1,b_1,\ldots,a_n,b_n)$, define the prefix sums
$s_0\defeq0$ and $s_j\defeq\sum_{\ell=1}^j w_\ell$
for $1\leq j\leq2n$.
Since $w\in E$ has even weight, $s_{2n}=0$.
Each prefix sum $s_j$ with $1\leq j<2n$ appears twice in
$\Phi(0,w)$, and the first and last coordinates are zero.
It follows that
$$
\wt(\Phi(0,w))
=2\bigl|\{j:1\leq j<2n,\ s_j=1\}\bigr|
$$
We claim that the number of nonzero prefix sums on the right
is even. Working over $\F$, we have
$$
\sum_{j=1}^{2n-1}s_j
=\sum_{i=1}^n(s_{2i-1}+s_{2i})
=\sum_{i=1}^n b_i=0
$$
The first (equality) step in the equation above uses the fact that $s_{2n}=0$.
For the second step, observe that consecutive prefix sums satisfy
$s_{2i-1}+s_{2i}=b_i$ (this is exactly the same trick as in our \cref{prefix-sum-construction} proof.)
The last step holds because $b\in P$ has even weight.
Thus, the number of nonzero prefix sums is even, and hence
$\wt(\Phi(0,w))\equiv0\pmod4$.

Both summands belong to the self-dual code $M$, so they are
orthogonal. By \cref{equation-weight-overlap}, their sum
has weight divisible by four as well.
This proves that $M$ is Type II.

We next check the embedding.
Since $n\geq d(P)>2H\geq m$, we have $N=4n\geq4m$,
so $\tau$ is well-defined.
Its image is doubly even and hence self-orthogonal.
Also, $\wt(\tau(x))=4\wt(x)\leq4m\leq H$.
Together with $d(M)>2H$, this bound gives
\cref{equation-typed-protection} by the triangle inequality.

All the hypotheses of \cref{metric-exchange} now hold.
Clauses~(\ref{generalized-main-clause-1}) -
(\ref{generalized-main-clause-3}) give the metric and recovery
properties, and Clause~(\ref{generalized-main-clause-4})
shows that $A$ is Type II.
Efficient completion and recovery follow from
Clause~(\ref{generalized-main-clause-6}).
The construction of $M$ and $\tau$ depends only on $m,H$ 
and is deterministic poly($N$)-time.
\end{proof}

%%%%%%%%%%%%%%%%%%%%%%%%%%%%%%%%%%%%%%%%%%%%%%%%%%%%%%%%%%%%%%%%%%%%

\end{document}